\documentclass{IEEEtran}
\IEEEoverridecommandlockouts
\usepackage{graphicx}
\usepackage{comment}
\usepackage{xcolor}
\usepackage{soul}
\usepackage{epstopdf}
\usepackage{amsmath}
\usepackage{amsthm}
\usepackage{relsize}
\usepackage{bm}
\usepackage{amssymb}
\usepackage{amsfonts}
\usepackage{url}
\usepackage{cite}
\usepackage{subfigure}
\usepackage{caption}
\usepackage{setspace}
\usepackage{float}
\usepackage{lipsum}
\title{\LARGE \bf RFID-Assisted Indoor Localization Using Hybrid Wireless Data Fusion}
\author{
\IEEEauthorblockN{Abouzar Ghavami, \textit{Member, IEEE}, and Ali Abedi, \textit{Senior Member, IEEE}
\thanks{The authors are affiliated with Electrical and Computer Engineering Dept., University of Maine, Orono, ME, USA. Emails: \{ \texttt{ghavamip@gmail.com, ali.abedi@maine.edu} \},
\newline This project was financially sponsored by Maine Economic Improvement Fund (MEIF) and University of Maine Office of the Vice President for Research.}
}}
\begin{document}
\maketitle
\newtheorem{Prop}{Proposition}
\newtheorem{Def}{Definition}
\newtheorem{Thm}{Theorem}
\newtheorem{Alg}{Algorithm}
\newtheorem{Lem}{Lemma}
\thispagestyle{empty}
\pagestyle{empty}
\begin{abstract}
Wireless localization is essential for tracking objects in indoor environments. Internet of Things (IoT) enables localization through its diverse wireless communication protocols. In this paper, a hybrid section-based indoor localization method using a developed Radio Frequency Identification (RFID) tracking device and multiple IoT wireless technologies is proposed. In order to reduce the cost of the RFID tags, the tags are installed only on the borders of each section. The RFID tracking device identifies the section, and the proposed wireless hybrid method finds the location of the object inside the section. The proposed hybrid method is analytically driven by linear location estimates obtained from different IoT wireless technologies. The experimental results using developed RFID tracking device and RSSI-based localization for Bluetooth, WiFi and ZigBee technologies verifies the analytical results.
\end{abstract}
\section{Introduction} \label{Sec:Intro}
In recent years, commercial wireless technologies such as Global Positioning System (GPS), Global System for Mobile communication (GSM), Bluetooth and WiFi (Wireless Fidelity) are widely used in personal communication devices such as smart phones and tablets. While other communication protocols are specifically used for data transmission, GPS is specifically designed to help us navigate on the surface of the earth \cite{G08}. However, due to the multipath effects and path loss of GPS signals through different building materials the current commercial GPS systems are unable to provide accurate indoor localization information \cite{NLD14, SWS13}. % with a free space horizontal accuracy of 7.8 meters at a $95$ percent confidence interval. 

Internet of Things (IoT) provides a communication platform for different wireless technologies including Bluetooth, WiFi, ZigBee, IPv6 Low-power Wireless Personal Area Network (6LoWPAN),  Low Power Wide Area Network (LoRaWAN), Bluetooth Low Energy (BLE), to communicate  through their designated base stations \cite{AAA17, XHL14}. There are various indoor localization methods using wireless technologies such as Bluetooth, WiFi, ZigBee, Radio Frequency Identification (RFID), infrared and acoustic waves. These localization methods use physical characteristics of the transmitted signals such as time of flight, angle of arrival and received signal strength indicator (RSSI) measurements, to estimate the position of the moving objects inside buildings \cite{LDBL07}. 

RSSI-based localization using individual IoT's wireless technologies is proposed in wireless localization studies \cite{YAM14, SS18, XQHY17, SWS13, AMH16}. Localization using WiFi RSSI information is improved through averaging selected maximum RSSI observations in \cite{XQHY17}. The performance of individual RSSI-based localization using BLE, WiFi, ZigBee and LoRaWAN protocols is compared through experiments and logarithmic distance approximations in \cite{SS18}. In our experimental results we also use RSSI-based localization using BLE, WiFi and ZigBee. %However, in order to maintain the accuracy of the measurements, the exact RSSI-based fingerprint locations and their unbiased mean estimators rather than RSSI logarithmic distance approximation is used in the experimental results of this paper. 
We combine the localization measurements obtained from these wireless protocols, and we obtain more accurate localization.

% IoT enables localization through its diverse wireless communication platform.

The concepts of combining different unconstrained estimators with smaller variances has been extensively studied in the literature \cite{GD59, RW74, M95, RKB13}. In this paper, a constrained method of combining different estimators with a proposed generalized distance function is studied. A localization method that combines different localization technologies utilizing prediction techniques is proposed in \cite{RKB13}. Gaussian probability kernel is used to approximate the probability density function of the RSSI at each finger-printed location. A multi-sensor data fusion approach that combines GPS and Inertial Measurement Units (IMU) is proposed in \cite{CDPV06}. Kalman filtering is applied simultaneously to GPS and IMU data. These methods require the covariance matrix of localization error for all the fingerprints in each individual localization method. A linear combination of estimates using RSSI level of WiFi and RSSI level of a stationary RFID device is used to find the optimal location estimate at each individual location in \cite{RKB13}. However, this paper proposes to use a developed portable RFID device to identify the exact location of the tags that are installed in the localization environment. The idea of using portable RFID device for emergency agents tracking is also proposed in \cite{RYTM07}. The paper proposes to use few RFID tags in the main enterance of the building to have a rough location of the person inside the buildig. A linear estimates for data fusion of light, acoustic and RSSI signals is proposed using confidence-weighted averaging fusion method in \cite{KH16}. In this paper, a data fusion method for indoor localization application is proposed that combines multiple estimators that are (i) constrained and (ii) minimizes the localization error for a generalized distance function. 

The generalized distance function that is proposed in this paper allows us to measure generalized distance error factors such as the Mean Square Error (MSE) and Mean Absolute Error (MAE) of the proposed hybrid method. In the case that MSE is the factor of interest, finding the optimal linear hybrid distance estimator is a constrained convex quadratic programming problem (QP). This constrained QP problem can also be solved through efficient and iterative primal-dual interior-point method \cite{W97}, iterative primal-dual coordinate ascent method \cite{B99} and iterative active set method \cite{GI83}.

In this paper, a hybrid localization method is proposed that estimates the locations using linear combination of the estimated locations of each individual localization methods. In this method, a unique weight is assigned to each estimated location such that the distance error of the final estimate with respect to the true location is minimized. As expected, it is analytically proven that the error distance of the hybrid approach is less than each individual localization method. Experimental results also demonstrate that hybrid estimates have less distance measurement error compared to individual localization methods. The rest of this paper is structured as follows. The hybrid localization method is proposed and extended in Section \ref{Sec:Hybrid}. In Section \ref{Sec:Numerical}, performance of the proposed hybrid localization method is compared with individual RSSI-based and RFID-based localization methods. We conclude in Section \ref{Sec:Conclude}.

\section{Proposed Hybrid Localization Method}\label{Sec:Hybrid}
Let us denote $\mathcal{N}$ as the set of $N$ different individual localization methods, numerated as $i$, $i \in \{1, ..., N\}$. Let $\mathbf{P}=(x, y, z)$ denote the position of the object in three dimensional space. Let us consider there are $M$ points with known locations, $\mathbf{P}^{j}$, $j \in \{1, ..., M\}$, used as the fingerprints in the localization environment.We assume that each individual localization method estimates the location of these fingerprints. The location of point $\mathbf{P}^{j}$ is estimated through $i^{th}$ localization method as $\hat{\mathbf{P}}^{j}_{i}(\mathbf{P}^{j})=(\hat{x}^{j}_{i}, \hat{y}^{j}_{i}, \hat{z}^{j}_{i})$.

Let us define $d_{\boldsymbol{\mathcal{V}}}(\mathbf{P}, \mathbf{Q}) \doteq ||\mathbf{P} - \mathbf{Q}||_{\boldsymbol{\mathcal{V}}}$ as the distance in vector function $\boldsymbol{\mathcal{V}}=\left(\mathcal{V}_x, \mathcal{V}_y, \mathcal{V}_z \right)$ of two points $\mathbf{P}$ and $\mathbf{Q}$, $\mathbf{P}=(x_p, y_p, z_p)$ and $\mathbf{Q}=(x_q, y_q, z_q)$, as
\begin{gather}
d_{\boldsymbol{\mathcal{V}}}\left(\mathbf{P},\mathbf{Q}\right) \doteq ||\mathbf{P} - \mathbf{Q}||_{\boldsymbol{\mathcal{V}}} = \boldsymbol{\mathcal{V}}\left(\mathbf{P}-\mathbf{Q}\right) \nonumber \\
= \mathcal{V}_x\left(x_p - x_q \right) + \mathcal{V}_{y}\left( y_p - y_q \right) + \mathcal{V}_z \left( z_p - z_q \right),
\label{Dv}
\end{gather}
where $\mathcal{V}_{u}\left(t\right)$, $u \in \{x, y, z\}$, is a strictly convex function of $t$. In special case, for $\mathcal{V}_u(t)=|t|^p$, $p>1$, the distance in vector function $\boldsymbol{\mathcal{V}}=\left(|t|^p, |t|^p, |t|^p\right)$ is the power $p$ of $p$-norm distance between $\boldsymbol{P}$ and $\boldsymbol{Q}$. For $p = 2$, the norm is the mean square error of the distance between the true point and its approximated measurement in the Eucledian space.

%Let us define $||\hat{\mathbf{P}}(\mathbf{P}^j) - \mathbf{P}^j||_{\boldsymbol{\mathcal{V}}^{j}} \doteq d_{\boldsymbol{\mathcal{V}}^{j}}(\hat{\mathbf{P}}(\mathbf{P}^j), \mathbf{P}^j)$ as the distance in vector function $\boldsymbol{\mathcal{V}}^{j} = \left(\mathcal{V}^{j}_x, \mathcal{V}^{j}_y, \mathcal{V}^{j}_z\right)$ of points $\hat{\mathbf{P}}(\mathbf{P}^j)$ and $\mathbf{P}^j$.
%\begin{gather}
%d_{\boldsymbol{\mathcal{V}}^{j}}(\hat{\mathbf{P}}(\mathbf{P}^j), \mathbf{P}^j) = \mathcal{V}^{j}\left(\hat{\mathbf{P}}\left(\mathbf{P}^{j}\right) - \mathbf{P}^{j}\right)  \nonumber \\
%= \mathcal{V}^{j}_{x}\left(\hat{x}^{j}-x^{j}\right) + \mathcal{V}^{j}_y\left(\hat{y}^{j}-y^{j}\right)+\mathcal{V}^{j}_z\left(\hat{z}^{j}-z^{j}\right)
%\label{normV}
%\end{gather}
%where $\mathcal{V}^{j}_{u}(\cdot)$, $u \in \{x,y,z\}$, is strictly convex function in $\mathbb{R}$. 
%For the strictly convex distance function $\mathcal{V}_u(t)=|t|^{p}$, $p > 1$, the distance function in (\ref{normV}) is the power $p$ of the $p$-norm distance of the points $\hat{\mathbf{P}}(\mathbf{P}^j)$ and $\mathbf{P}^j$. For $p = 2$, the norm is the mean square error of the distance between the true point and its approximated measurement in the Eucledian space.

%Let each method $m_i$ identify the position of the object in the space as . 
The hybrid localization method to find the optimal estimated location of point $\mathbf{P}$, $\hat{\mathbf{P}}(\mathbf{P})$, is proposed as the linear combination of different localization methods:
\begin{equation}
\hat{\mathbf{P}}\left(\mathbf{P}\right)=\left(\sum_{i=1}^{N}\alpha_{x, i} \hat{x}_{i}, \sum_{i=1}^{N}\alpha_{y, i} \hat{y}_{i}, \sum_{i=1}^{N}\alpha_{z, i} \hat{z}_{i}\right)
\label{LinEst}
\end{equation}
\begin{equation}
\begin{gathered}
\text{s.t.    } \sum_{i=1}^{N}\alpha_{u, i}=1, \hspace{5mm} \forall u \in \{x,y,z\} \\
\alpha_{u, i}\geq 0, \hspace{5mm} \forall i \in \{1, ..., N\}, \forall u \in \{x,y,z\}.
\end{gathered}
\label{Cons}
\end{equation}
where $\boldsymbol{\alpha}_{i} = (\alpha_{ x, i}, \alpha_{y, i}, \alpha_{z, i})$ is the vector of the coefficients for taking into account each individual coordinates of $x$, $y$ and $z$ that are estimated by localization method $i$.

Let $\boldsymbol{\alpha}_{u}=(\alpha_{u, 1}, ..., \alpha_{u, N})$ denote the coefficient vector of different localization methods in direction $u$, $u \in \{x, y, z\}$. Let $\mathcal{A}_u=\{\boldsymbol{\alpha}_{u}| \alpha_{u, i} \geq 0, \sum_{i=1}^{N}\alpha_{u, i}=1\}$ denote all the feasible sets of coefficient vectors in direction $u$. $\mathcal{A}_u$ can be shown to be convex and closed, hence it is a compact set. Let $\boldsymbol{\alpha} = (\boldsymbol{\alpha}_{x},\boldsymbol{\alpha}_{y}, \boldsymbol{\alpha}_{z})$ denote the coefficient vector of all localization methods. 
Let $\mathcal{A}=\{\boldsymbol{\alpha}|\alpha_{u, i} \geq 0: \sum_{i=1}^{N}\alpha_{u, i}=1, \forall u \in \{x,y,z\}\}$ denote the set of feasible vectors $\boldsymbol{\alpha}$ that satisfy the constraints in (\ref{Cons}). We have $\mathcal{A}=\mathcal{A}_{x} \times \mathcal{A}_{y} \times \mathcal{A}_{z}$. As $\mathcal{A}_{u}$, $u \in \{x, y, z\}$, is compact, $\mathcal{A}$ is also a compact set. 

%Let us denote $d_p\left(\hat{\mathbf{P}}(\mathbf{P}^j), \mathbf{P}^j\right)$ as the $p$-norm distance of points $\hat{\mathbf{P}}(\mathbf{P}^j)$ and $\mathbf{P}^j$. 
%\begin{gather}
%d^{p}_{p}\left(\hat{\mathbf{P}}(\mathbf{P}^j)-\mathbf{P}^j\right)
%= d_{|t|^p}\left(\hat{\mathbf{P}}(\mathbf{P}^j)-\mathbf{P}^j\right) \nonumber \\
%= ||\hat{\mathbf{P}}(\mathbf{P}^j)-\mathbf{P}^j||^{p}_p 
%= |\hat{x}^{j}-x^{j}|^p + |\hat{y}^{j}-y^j|^p + |\hat{z}^{j}-z^j|^p.
%\label{normP}
%\end{gather}
The optimal coefficient vector, $\boldsymbol{\alpha}^{*}$, minimizes the error distance of the true location of the points and their estimated locations in vector function $\boldsymbol{\mathcal{V}}^{j}$, $j=1,..., M$ as in the following:
\begin{gather}
\boldsymbol{\alpha}^{*}=\arg\min_{\boldsymbol{\alpha} \in \mathcal{A}}\frac{1}{M}\sum_{j=1}^{M}d_{\boldsymbol{\mathcal{V}}^{j}}(\hat{\mathbf{P}}(\mathbf{P}^j), \mathbf{P}^j) \nonumber \\
%= \arg\min_{\boldsymbol{\alpha} \in \mathcal{A}} \sum_{j=1}^{M}||\hat{\mathbf{P}}(\mathbf{P}^j)-\mathbf{P}^j||_{\boldsymbol{\mathcal{V}}^{j}} \nonumber \\
=\arg\min_{\boldsymbol{\alpha} \in \mathcal{A}}\sum_{u \in \{x, y, z\}}\sum_{j=1}^{M}\mathcal{V}_u^{j}\left(\sum_{i=1}^{N}\alpha_{u, i}\hat{u}^{j}_{i}-u^{j}\right).
\label{alphaOpt}
\end{gather}
As $\boldsymbol{\alpha}_u$, $u \in \{x,y,z\}$, are considered to be independent from each other, we have:
\begin{equation}
\boldsymbol{\alpha}_{u}^{*}=\arg\min_{\boldsymbol{\alpha}_u \in \mathcal{A}_{u}} \sum_{j=1}^{M}\mathcal{V}^{j}_{u}\left(\sum_{i=1}^{N}\alpha_{u, i}\hat{u}^{j}_{i}-u^{j}\right).
\label{alphau}
\end{equation}
Let us denote
\begin{equation}
f_{u}(\boldsymbol{\alpha}_{u})=\sum_{j=1}^{M}\mathcal{V}^{j}_{u}\left(\sum_{i=1}^{N}\alpha_{u, i} \hat{u}^{j}_{i}-u^{j}\right), \hspace{3mm} u \in \{x, y, z\}.
\label{falpha}
\end{equation}
Let us denote $\mathbf{u}_{i}^{j}=\left( \hat{u}_{i}^{1}, ..., \hat{u}_{i}^{M}\right)$ as the location information vector of localization method $i$ in direction $u$ and $\mathbf{U}_{M \times N}$ as the matrix of location information with $U_{ij}= \hat{u}_{i}^{j}$. 

Let $\mathbf{C}$ denote the \emph{information location correlation matrix} that is defined as
\begin{gather}
\mathbf{C} = \mathbf{U}^{T} \mathbf{U},
\label{CorrelationMatrix}
\end{gather}
where $C_{ij} = \sum_{m=1}^{M} \hat{u}^{m}_{i} \hat{u}_{j}^{m}$.

It is assumed that all of the localization information vectors are linearly independent from other vectors. Otherwise the linearly dependent vectors are redundant, and they are removed from the localization information vectors set.
\begin{Lem}
The \emph{information location correlation matrix}, $\mathbf{C}$, is a positive definite matrix.
\label{CPD}
\end{Lem}
\begin{proof}
See Appendix A.
\end{proof}
 
Let us denote $\hat{u}^{j}_{\max}=\max_{1 \leq i \leq N} \hat{u}_i^{j}$ and $\hat{u}^{j}_{\min}=\min_{1 \leq i \leq N} \hat{u}_i^{j}$. As $\alpha_{u, i}$, $\forall i, u$, satisfy the constraints in (\ref{Cons}), we have: 
\begin{equation}
\hat{u}_{\min}^{j}-u^j < \sum_{i=1}^{N} \alpha_{u, i} \hat{u}_{i}^{j} - u^{j}< \hat{u}_{\max}^{j} - u^{j}.
\end{equation}
Let us denote 
\begin{equation}
B_{u}^{j}=\max_{t \in [\hat{u}_{\min}^{j}-u^j, \hat{u}_{\max}^{j} - u^{j}]} {\mathcal{V}_{u}^{j}}'' \left(t\right),
\end{equation}
and 
\begin{equation}
b_{u}^{j}=\min_{t \in [\hat{u}_{\min}^{j}-u^j, \hat{u}_{\max}^{j} - u^{j}]} {\mathcal{V}_{u}^{j}}'' \left(t\right).
\end{equation}
Fom strongly convexity of $\mathcal{V}_{u}^{j}(x)$, $\mathcal{V}_{u}^{j}(x) '' > 0$, thus: $b_{u}^{j} > 0$. Let us denote
\begin{equation}
L^{\max} = \max_{i, t} \left \{ \sum_{j=1}^{M} \left | \hat{u}^{j}_{i} \hat{u}_{t}^{j} \right | B_{u}^{j} \right \}
\end{equation}
and
\begin{equation}
l^{\min} = \min_{i} \left\{ \sum_{j=1}^{M} \left( \hat{u}_{i}^{j} \right)^2 b_{u}^{j} \right \}.
\end{equation}
As $\exists j: \hat{u}_{i}^{j} \neq 0$, and $\forall j: b_{u}^{j}>0$, thus: $l^{\min}>0$. 

The objective function in (\ref{alphaOpt}) is convex in $\boldsymbol{\alpha}$ that let us to solve the optimization problem in (\ref{alphaOpt}) using gradient projection method with constant step size $\beta_u \in (0, \frac{2}{N L^{\max}})$ \cite{B99}. The optimal solution is derived through the iterative update procedure
\begin{equation}
\boldsymbol{\alpha}_{u}^{k+1}=\mathcal{P}_{\mathcal{A}_{u}}\left[\boldsymbol{\alpha}_{u}^{k} - \beta_{u} \nabla f_{u}(\boldsymbol{\alpha}_{u}^k)\right], \hspace{5mm} k \in \mathbb{Z}^{+}.
\label{alphak}
\end{equation}
where $\mathcal{P}_{\mathcal{A}_{u}}[\boldsymbol{\alpha}_u]$ is the projection of the vector $\boldsymbol{\alpha}_{u}$ over the hyperplane $\mathcal{A}_{u}$. The vector $\boldsymbol{\alpha}_{u}^{0}$ is an arbitrary vector in $\mathcal{A}_{u}$. %As $\hat{u}_{i}^{j}$ and $u^j$ are limited values, for a specific set of values $u_{i}^{j}$ and $u^{j}$, there exists a value $B_u>0$ where $B_u \geq \left(\sum_{i=1}^{N}\alpha_{i,u}\hat{u}_{i}^{j}-u^{j}\right)^{p-2}$, $\forall u_{i}^{j}, u^{j}$.

\begin{Thm}
If $\mathcal{A}_u$ is not empty, for $0 < \beta_u <\frac{2}{N L^{\max}}$, $u \in \{x, y, z\}$, $f_{u}(\boldsymbol{\alpha}^k)$ decreases monotonically as $k$ increases, and  $f_{u}^*=\lim_{k \rightarrow \infty}f(\boldsymbol{\alpha}_{u}^{k})$ minimizes $f_u(\boldsymbol{\alpha}_u)$ over $\mathcal{A}_u$.
\label{Thm1}
\end{Thm}
\begin{proof}
See Appendix B. 
\end{proof}
The optimal projected coefficient vector, $\boldsymbol{\alpha}_u^{k+1}$, at each iteration of the gradient projection method in (\ref{alphak}) can be found as described next.
%**********************
%** Projection algorithm **
%**********************
\subsection{Optimal Projection Algorithm} \label{OptimalProjectionAlgorithm}
Let us denote $\boldsymbol{\omega}_{u}^{k}=\nabla f_u(\boldsymbol{\alpha}_u^{k})$, i.e. $\omega_{u, i}^{k}=\frac{\partial f_{u}(\boldsymbol{\alpha}_u^{k})}{\partial \alpha_{u, i}}$, $i=1, ..., N$. As $\mathcal{A}_u$ is closed and convex, based on the Projection Theorem in \cite{B99}, there is a unique projection vector $\boldsymbol{\alpha}_u^{k+1} \in \mathcal{A}_u$ that minimizes $||\boldsymbol{\alpha}_u - \left(\boldsymbol{\alpha}_u^k - \beta_{u} \boldsymbol{\omega}_u^{k})\right)||_2^2$, $\forall \boldsymbol{\alpha}_u \in \mathcal{A}_u$, as in (\ref{alphak}). Let us denote $\boldsymbol{\alpha}_u^{*}=\boldsymbol{\alpha}_{u}^{k+1}$. The projection problem in (\ref{alphak}) is rewritten as
\begin{gather}
\boldsymbol{\alpha}_u^{*}=\arg\min_{\boldsymbol{\alpha} \in \mathcal{A}_u} ||\boldsymbol{\alpha}_u - \left(\boldsymbol{\alpha}_u^{k}-\beta_{u} \boldsymbol{\omega}_u^{k}\right)||_2^2 \nonumber \\
=\arg\min_{\boldsymbol{\alpha} \in \mathcal{A}_u} \sum_{i=1}^{N}\left(\alpha_{u, i} + \beta_{u} \omega_{u, i}^{k} - \alpha_{u, i}^{k}\right)^2.
\label{Opt2}
\end{gather}
Let us define $[x]^+ = \max(x, 0)$, and denote $c^{k}_{u, i}=\beta_{u} \omega_{u, i}^{k} - \alpha_{u, i}^{k}$. The least-square optimization problem in (\ref{Opt2}) is solved through following Lemma.%\lq\lq water filling'' approach.
\begin{Lem}
The optimal solution of (\ref{Opt2}) is uniquely derived as $\alpha_{u, i}^*=[\lambda_u^* - c^{k}_{u, i}]^{+}$, $i = 1, ..., N$, where $\lambda_u^*$ is a unique constant in the interval $\Lambda_u=[\min_{i}c_{u, i}^{k}, 1+\min_{i} c_{u, i}^{k}]$.
\label{Lem1}
\end{Lem}
\begin{proof}
See Appendix C.
\end{proof}
The constant value $\lambda_u^*$ is found through bisection method in $O(N \log(\frac{1}{\epsilon}))$ algorithm running time with maximum predetermined error of $\epsilon > 0$. Its exact value can also be found in $O(N\log(N))$ algorithm running time as in the following Lemma.
\begin{Lem}
Let $m_1 \leq ... \leq m_N$ denote the ordered sequence of $\{c_{i,u}^{k}\}_{i=1}^{N}$, and let us set $m_{N+1}=m_1 + 1$. There exists a unique value $n$, $1 \leq n \leq N$ such that the constant value for $\lambda^{*}_{u}$ in Lemma \ref{Lem1} is derived as
\begin{equation}
m_n \leq \lambda_u^{*}=\frac{1 + \sum_{i=1}^{n}m_i}{n} \leq m_{n+1}.
\label{lambdau}
\end{equation}
%where $n$ is the smallest value that satisfies (\ref{lambdau}).
\label{Lem2}
\end{Lem}
\begin{proof}
See Appendix D.
\end{proof}
It takes $O(N \log(N))$ algorithm running time to order the sequence $\{c_{i,u}^{k}\}_{i=1}^{N}$ as $\{m_i\}_{i=1}^{N}$. It also takes $O(N)$ algorithm running time to calculate $\lambda_u^{*}$ for all $n$ as in (\ref{lambdau}). It results in total number of $O(N \log(N))$ calculations to find $\lambda_{u}^{*}$.

Let us denote 
\begin{equation}
q = \max\{|1 - \beta_{u} N l^{\min}|, |1 - \beta_{u} N L^{\max}|\}
\label{q}
\end{equation}
From $0 < \beta_u < \frac{2}{N L^{\max}}$ and $0 < l^{\min} \leq L^{\max}$, we have: $0 < q < 1$. Let the update process in (\ref{alphak}) start with $\boldsymbol{\alpha}_{u}^{0}$, $\boldsymbol{\alpha}_{u, i}^0 = \frac{1}{N}$, $i = 1,..., N$.
\begin{Thm}
The iterative update process in (\ref{alphak}) converges to a solution in the $\epsilon$-neighborhood of the optimal solution with linear convergence rate of $q$ in $O \left( N \log \left( N \right) \log_{\left( \frac{1}{q} \right) } \left( \frac{\sqrt{1 - \frac{1}{N}}}{ \mathlarger{ \mathlarger{ \epsilon } } } \right) \right)$ number of calculations.
\label{Thm2}
\end{Thm}
\begin{proof}
See Appendix E.
\end{proof}
For minimizing the absolute error distance measurement, i.e. $1$-norm error distance, it is possible to use the proposed gradient decent method with the strictly convex function $\mathcal{V}_{u}^{j}(t)=|t|^{1+\epsilon}$, where $\epsilon$ is set to a sufficiently small constant value, $\epsilon \ll 1$.

In the experimental results, the MSE of the true locations and their estimated values are minimized. In this case, the distance function is set as $\mathcal{V}_{u}^{j}(t)=t^{2}$, $\forall u, j$. The distance minimization problem in (\ref{alphaOpt}) is converted to a constrained convex quadratic programming problem (QP) that may be solved through efficient and iterative primal-dual interior-point method \cite{W97}, iterative primal-dual coordinate ascent method \cite{B99} and iterative active set method \cite{GI83}.

\subsection{Section-based Hybrid Localization}\label{Sec:Sectioning}
In order to improve the accuracy of the location estimation, the indoor localization environment is divided into different mutually exclusive sections. The coefficients for each wireless technology is determined for each section. Let us consider the indoor environment is divided into $S$ sections, numbered as $\mathcal{S}=\{1, ..., S\}$. The section of the mobile object is determined based on its estimated location using the proposed method in Section \ref{Sec:Hybrid}. Let us denote $\hat{S}=S(\hat{\mathbf{P}}\left(\mathbf{P}\right))$ as the estimated section that contains the location $\mathbf{P}$. The estimated location $\hat{\mathbf{P}}\left(\mathbf{P}\right)$ is derived from (\ref{LinEst}) by solving (\ref{alphaOpt}). Let us denote $\alpha_{u, i, \hat{S}}$ as the weight of location estimation using wireless technology $i$, $i \in \{1, ..., N\}$ in direction $u$, $u \in \{x, y, z\}$, when point $\mathbf{P}$ is estimated to be located in section $\hat{S}$. Let us denote $\boldsymbol{\alpha}_{u,\hat{S}}=(\alpha_{u, 1,\hat{S}}, ..., \alpha_{u, N,\hat{S}})$ as the coefficient vector of different localization methods in direction $u$ in section $\hat{S}$. Let $\mathcal{A}_{u,\hat{S}}=\{\boldsymbol{\alpha}_{u,\hat{S}}|\alpha_{u, i, \hat{S}} \geq 0, \sum_{i=1}^{N}\alpha_{u, i, \hat{S}}=1\}$ denote all the feasible sets of coefficient vectors in direction $u$ in the estimated section $\hat{S}$. Let $\boldsymbol{\alpha}=(\boldsymbol{\alpha}_{1, 1}, ..., \boldsymbol{\alpha}_{N,S})$ denote the coeffiecient vector for all localization methods in all defined sections. Using the same approach in (\ref{alphaOpt}) and (\ref{alphau}), the optimal coefficient vector, $\boldsymbol{\alpha}^{*}$, minimizes the error distance in vector function $\boldsymbol{\mathcal{V}}^j=\left(\mathcal{V}^{j}_x, \mathcal{V}^{j}_y, \mathcal{V}^{j}_z\right)$, where $\mathcal{V}^{j}_u$, $u \in \{x, y, z\}$, are strictly convex functions in $\mathbb{R}$. Each coefficient vector $\boldsymbol{\alpha}_{u, \hat{S}}^{*}$ is derived as in the following:
\begin{gather}
\boldsymbol{\alpha}_{u,\hat{S}}^{*}=\arg\min_{\boldsymbol{\alpha}_{u,\hat{S}} \in \mathcal{A}_{u, \hat{S}}}\sum_{j: \hat{\bold{P}}(\bold{P}^{j}) \in \hat{S}}\mathcal{V}^{j}_{u}\left(\sum_{i=1}^{N}\alpha_{u, i, \hat{S}}\hat{u}_{i,\hat{S}}^{j}-u^j\right), \nonumber \\
\forall u \in \{x, y, z\}, \forall \hat{S} \in \mathcal{S}.
\label{alphaS}
\end{gather}
The optimal coefficient vector for section-based localization is derived using gradient projection method following the same approach proposed in Section \ref{Sec:Hybrid}. %Note that the number of sections is limited to the number of measurement data points.

\section{Experimental Results}\label{Sec:Numerical}
In this section, experimental and simulation results based on RSSI measurements for Bluetooth, WiFi and ZigBee protocols, are presented. Three wireless transmitters are set up: (i) the CC2650 Bluetooth Sensortag (Texas Instruments Inc., Dallas, TX), (ii) the CC2650 ZigBee Sensortag (Texas Instruments Inc., Dallas, TX) and (iii) WiFi dongle (CanaKit Inc., Vancouver, Canada), where WiFi dongle is connected to a Raspberry Pi mini computer (Raspberry Pi Inc., UK). The two Sensortag transmitters and WiFi dongle transmitter on the Raspberry Pi are attached to a quad-cane (Medline Industries, Mundelein, IL), for measuring the RSSI at different distances as the quad-cane is moved.

Two different sets of multi-section experimets are conducted with and without RFID localization. In RFID localization experiments, a UMaine developed UHF RFID tracker is attached under the quad-cane that localizes the position of the UHF RFID tags embedded under the floor carpet. The UHF RFID tracker is equipped with low weight portable 6800 mAh 5v battery that allows it to track an UHF RFID tag within short range of $30$ cm. The RFID tracker sends RF transmission signal per second. Although the RFID tracking accuracy is less than $30$ cm, using RFID tags are costly. In order to avoid covering the whole environment with RFID tags, the environment is divided into different sections, and the RFID tags are only located at the border of each section to identify the section that tracker is located in. The tracker is localized inside the section using the RSSI information of the low-cost Bluetooth, WiFi and ZigBee transmitters mounted on the quad-cane. The experiments are conducted based on creating a fingerprint data set with measuring RSSI at specific locations in the eastern hallway of the Engineering and Science Research Building (ESRB) at the University of Maine, Orono, ME. 

The experiments are set up for $x$ direction along the $60$ m hallway where the distance between two adjacent fingerprint locations is $0.915$ m. In order to have a fair comparison between the accuracy of each transmission protocol, the Bluetooth, WiFi and ZigBee receivers are located in the same location at the beginning of the hallway. We assume that the RSSI measurement variations are not high between two adjacent fingerprint locations in the training set. Due to the accesibility of all indoor locations, in the case that there exists a point between two specific adjacent fingerprint locations that large RSSI measurement variation is observed, that point is added to the fingerprint data set. Figure \ref{RSSI} shows the RSSI variations with distance of the tracking device from three WiFi, ZigBee and BLE base stations located at the same position in the entrance of the hallway.

\begin{figure}
\includegraphics[width = 0.5 \textwidth]{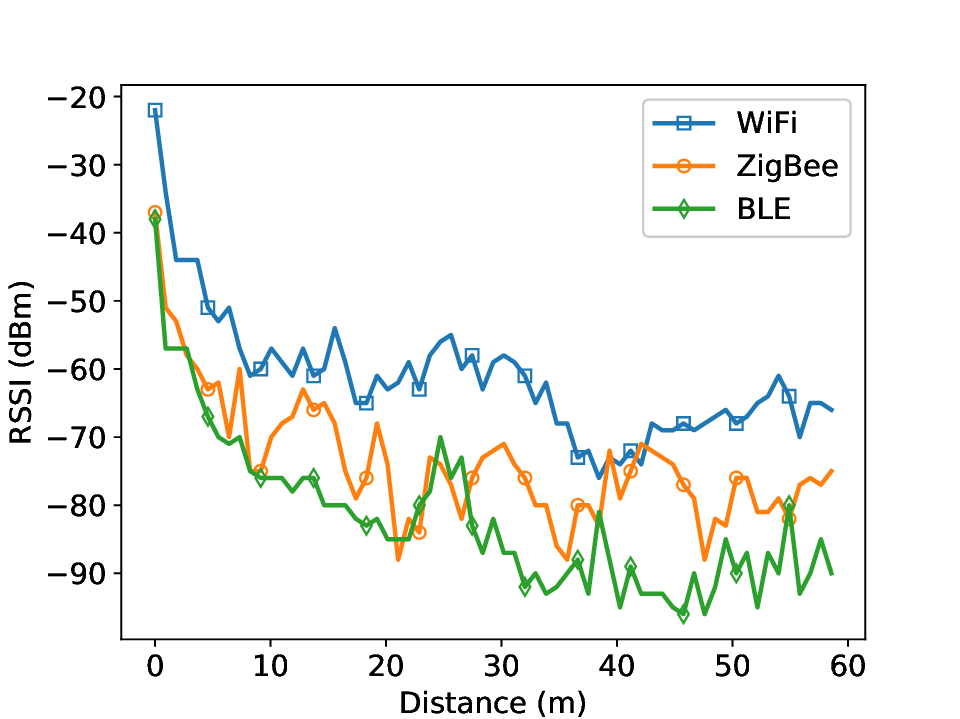}
\caption{RSSI variations of WiFi, ZigBee and BLE signals vs distance of the tracker from the location of the base stations.}
\label{RSSI}
\end{figure} 

Two different scenarios are considered for the experimental test: (i) Environment without measurement noise and (ii) Environment with measurement noise. In the first scenario, we consider an experimental environment in which the RSSI measurement for each location point is fixed at any time of RSSI measurement, and it does not change with time. In this case, the training set and test set data to find the optimal coefficients are the same. In order to increase the accuracy of the measurements, as each point of the environment is accessible, we do the measurements for as much as possible points in the environment. In this scenario we have performed measurements for the consecutive points each is located in $30$ cm distance from its adjacent points. It is recommended to do the measurements for all possible points in the environment such that the variation of measurements between two adjacent points is limited to a certain predetermined value. Figures \ref{BTHybrid}, \ref{WiFiHybrid} and \ref{ZBHybrid} demonstrate the distance estimation based on mean estimation using Bluetooth, WiFi and ZigBee RSSI information, respectively. The estimated location using the proposed hybrid method is also plotted in these figures. As seen in these figures, the proposed hybrid localization method achieves a localization curve that is on average closer to the real distance line in comparison to each individual wireless protocol.
\begin{figure}%[H]
\begin{minipage}{0.5\textwidth}
        \centering
        \subfigure[a][Bluetooth]{\includegraphics[width = \textwidth]{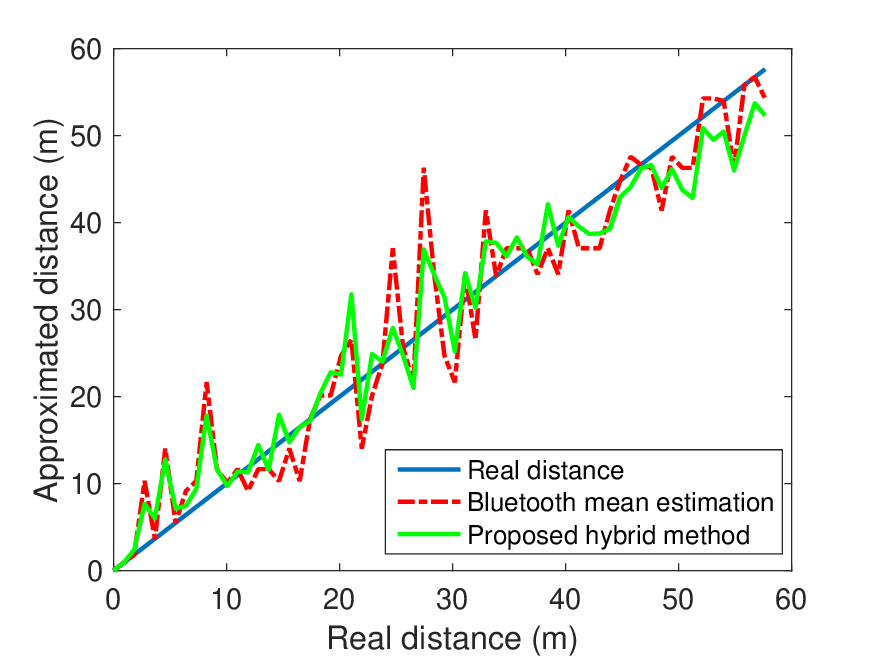}
        %\caption{Experimental measurment and theoretical approximation (\ref{RSSIapprox}) of RSSI for Bluetooth signal vs. distance from the transmitter.}
       \label{BTHybrid}
       }
   \end{minipage}
    \begin{minipage}{0.5\textwidth}
        \centering
            \subfigure[b][WiFi]{\includegraphics[width = \textwidth]{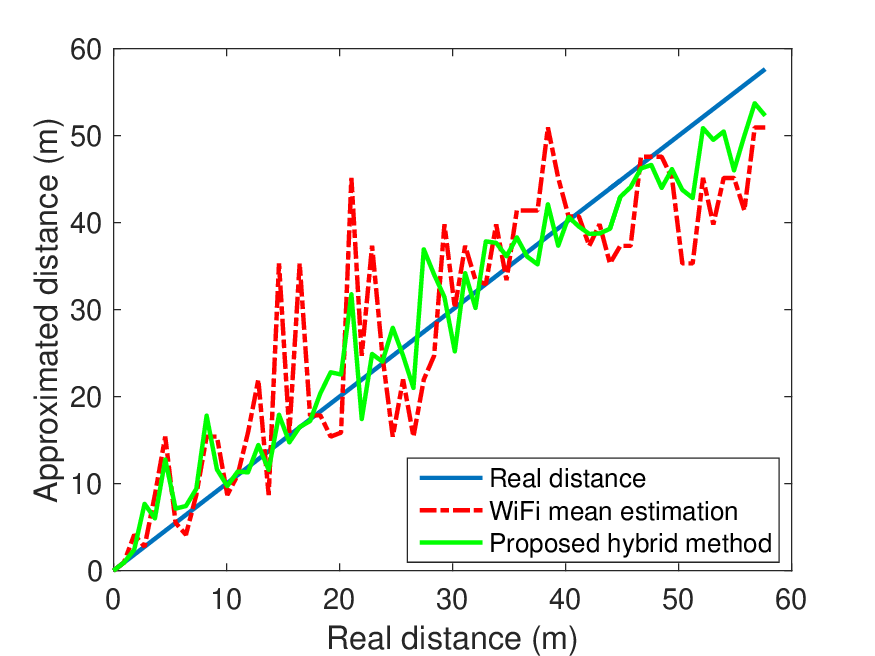}
           %\caption{Experimental measurment and theoretical approximation (\ref{RSSIapprox}) of RSSI for WiFi signal vs. distance from the transmitter.}
          \label{WiFiHybrid}
          }
    \end{minipage}
    \begin{minipage}{0.5\textwidth}
           \centering
           \subfigure[c][ZigBee]{
           \includegraphics[width = \textwidth]{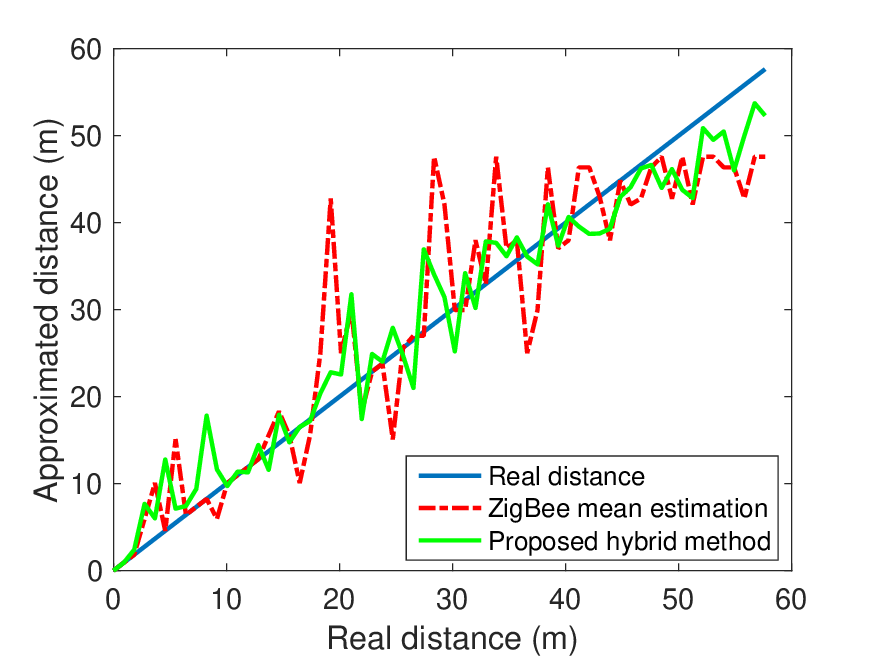}
         \label{ZBHybrid}
          }
     \end{minipage}
\caption{Distance approximation based on mean estimation using (a) Bluetooth, (b) WiFi and (c) ZigBee RSSI information and proposed hybrid method combining all three estimates in the first scenario.}
\end{figure}
Table \ref{Table1} summarizes the mean square error ($p=2$) of the estimated distance with respect to the real distance for each individual localization methods, the proposed hybrid method, proposed section-based method and RFID-based section-based method. The proposed hybrid method finds the MSE estimated distance with combining the Bluetooth, WiFi and ZigBee RSSI-based distance estimations. The proposed hybrid section-based method uses the MSE estimated distance obtained from the proposed hybrid method to estimate the section where the moving object is located in. The method applies equation (\ref{alphaS}) to find the MSE estimated location based on the estimated section in the first use of proposed hybrid localization method. 
In RFID-based section-based method, the sections are determined without error with reading the RFID tags that are implemented on the borders of the sections, and the location is estimated using equation (\ref{alphaS}) for each section.

\begin{table*}[t]
\caption{\doublespacing Mean square error (m) of different localization methods for different distance ranges of measurements in the first scenario}
\centering
\begin{tabular}{c c c c c c c}
\hline \hline
Distance & Bluetooth & WiFi & ZigBee 
& \multicolumn{1}{p{2.5cm}}{\centering Proposed hybrid method with one section and no RFID} 
& \multicolumn{1}{p{3cm}}{\centering Proposed two-level section-estimation method with three sections with no RFID} 
& \multicolumn{1}{p{2.5cm}}{\centering Proposed hybrid method with three sections using RFID sectioning} \\
\hline\\
20 m & 3.16 m & 2.01 m & 2.23 m & 1.86 m & 0.83 m & 0.23 m\\ \\
40 m & 4.42 m & 4.18 m & 4.01 m & 2.77 m & 1.58 m & 0.98 m\\ \\
60 m & 5.12 m & 8.07 m & 6.73 m & 3.98 m & 2.06 m & 1.73 m\\
\end{tabular}
\label{Table1}
\end{table*}
%\normalsize
As it is expected, the proposed RFID-based section-based method has the least mean square error with respect to other localization methods represented in Table \ref{Table1}. For $60$ m distance localization, it improves the localization MSE by $66\%$, $79\%$ and $74\%$ percent with respect to Bluetooth, WiFi and ZigBee standards, respectively. In this range, the proposed two layer hybrid section-based method improves the MSE by $60\%$, $75\%$ and $70\%$, with respect to Bluetooth, WiFi and Zigbee standards, respectively. For $60$ m distance localization, the optimal coefficient vector, $\boldsymbol{\alpha}$, is derived as $\boldsymbol{\alpha}=(0.58, 0.17, 0.25)$ for Bluetooth, WiFi and ZigBee standards, respectively. It shows that the proposed hybrid method improves the localization accuracy in terms of mean square distance error by $21\%$, $50\%$, and $40\%$ with respect to each individual localization method using RSSI-based Bluetooth, WiFi and ZigBee standards, respectively.

Although the proposed RFID-based section-based method has the best MSE localization performance, the experimental results show the proposed two-layer section-based method without RFID tracking does not significantly degrade the localization performance.

In the second scenario, in order to increase the robustness of the proposed method to measurement noise, the location dataset is divided into two mutual exclusive sets: (i) training set and (ii) test set. The coefficients of each localization method for each section is derived using the training set and the results are applied to the locations in the test set to find the %mean square error (MSE) and
mean absolute error (MAE) distance. % For MSE, the distance function in (\ref{falpha}) is set as $\mathcal{V}_{x}(t)=t^2$ and 
For MAE, the distance function in (\ref{falpha}) is set to $\mathcal{V}_{x}(t)=|t|^{1+\epsilon}$ where $\epsilon=0.0001$. In this experiment, $70$ percent of data is randomly chosen as the training set and the remainig $30$ percent is used as the test set. The experiment is repeated for $1000$ times and the average MAE of the proposed method on the test set data is shown in figure \ref{MAE}. 

For each experiment, the environment is divided into equally length sections with UHF RFID tags located on the borders of the sections. For only RFID method, the estimated location is selected as the middle point of the section. As it is seen in figure \ref{MAE}, the values of %MSE and 
MAE plots are very close to each other. This means that changing in distance function does not improve the desired localization accuracy for this scenario. As it is seen in this figure, the proposed RFID and RSSI-based hybrid method has better performance of RFID and each individual RSSI-based localization method. However, when the environment is divided to more than three sections, due to the smaller size of the sections, high level measurement inaccuracies caused by multipath fading inside the building and separation of training and test datasets in this scenario, the proposed hybrid method does not improve the performance of RFID only method which estimates the location as the middle point of the identified section.
%\begin{figure}
 %   \begin{minipage}{0.5\textwidth}
 %       \centering
%        \subfigure[a][MSE]{\includegraphics[width = \textwidth]{./Figures/MSEAllTrainingTest.eps}
        %\caption{Experimental measurment and theoretical approximation (\ref{RSSIapprox}) of RSSI for Bluetooth signal vs. distance from the transmitter.}
%       \label{MSE}
%       }
%   \end{minipage}
%   \begin{minipage}{0.5\textwidth}
%        \centering
%            \subfigure[b][MAE]{\includegraphics[width = \textwidth]{./Figures/MAEAllTrainingTest.eps}
%            \label{MAE}
%          }
%   \end{minipage}

%\caption{(a) Mean square error (MSE) distance and (b) mean absolute error (MAE) distance vs. number of environment sectioning using RFID tags and developed device for different proposed methods in the second scenario.}
%\label{MSEMAE}
%\end{figure}

\begin{figure}
\centering
\includegraphics[width = 0.5 \textwidth]{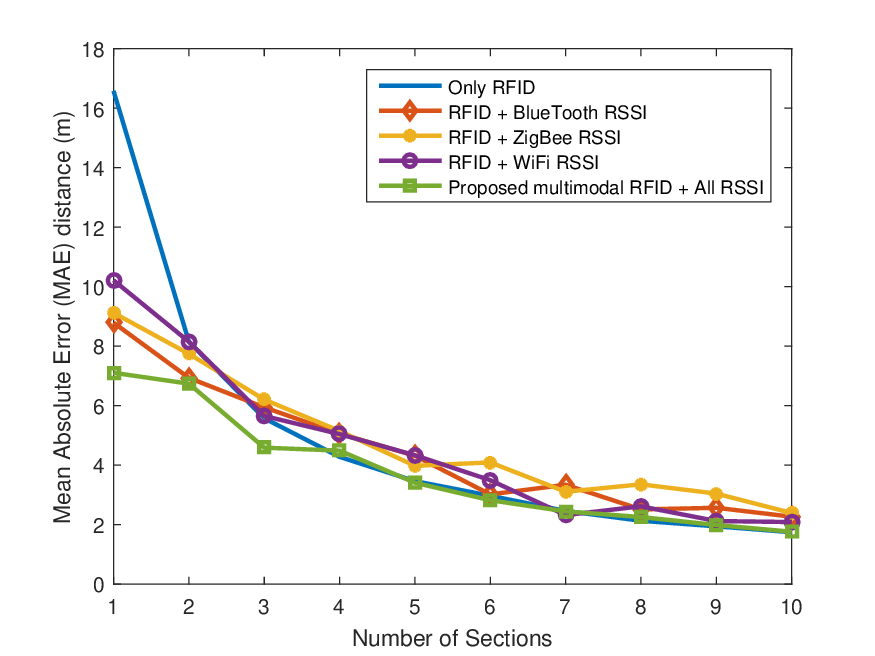}
\caption{Mean absolute error (MAE) distance vs. number of environment sectioning using UHF RFID developed tracker for different proposed methods in the second scenario.}
\label{MAE}
\end{figure}

\section{Concluding Remarks and Future Work}
\label{Sec:Conclude}
In this paper, a hybrid localization method is proposed that combines different IoT wireless localization methods to achieve a more accurate estimated location. The proposed hybrid method is based on the linear combination of the position estimates of each individual localization protocol such that a genaralized distance error of the hybrid location estimation, including the mean square error and mean absolute error, is minimized. Minimizing the mean square or mean absolute distance error ensures that the estimated localization curve, on average basis of the number of the measured estimated locations, is closer to the true distance curve. In other words, energy of error is minimized. A promising direction to decrease the mean square error even more, is proposed which involves partitioning of the environment into multiple sections such that the coefficient vectors for each section is different with other sections. The experimental results using IoT wireless technologies such as Bluetooth, WiFi and ZigBee protocols and developed RFID tracking device show that the proposed hybrid localization method achieves less mean square and mean absolute error in comparison with each individual RSSI-based localization methods. For wider range of experimental studies, the proposed hybrid method can also be applied to two and three dimensional localizations that require multiple base stations for each IoT wireless technology.

%In experimental studies, the following wireless protocols: Bluetooth, WiFi and ZigBee are used to estimate the distance based on RSSI measurements. The weights are analytically defined for each localization protocol to minimize the mean square error of the distance. 

\section{Appendix}
\subsection{Proof of Lemma \ref{CPD}:}
\begin{proof}
As we assumed all the localization information vectors are linearly independent, thus the equation $\mathbf{U}\mathbf{x} = 0$ has the unique solution $\mathbf{x}=\mathbf{0}$. For any vector $\mathbf{x}$, we have:
\begin{gather}
\mathbf{x}^{T} \mathbf{C} \mathbf{x} = \mathbf{x}^{T} \mathbf{U}^{T} \mathbf{U} \mathbf{x} \nonumber \\
= ||\mathbf{U} \mathbf{x}||_{2}^{2} \geq 0,
\end{gather}
The equality happens iff $\mathbf{x} = \mathbf{0}$. Thus, $\mathbf{C}$ is a positive definite matrix.
\end{proof}
\subsection{Proof of Theorem \ref{Thm1}:}
\begin{proof}
The partial derivative of $f_{u} \left( \mathbf{p} \right)$, $\mathbf{p} \in \mathcal{A}_{u}$, with respect to $\alpha_{u, i}$ is derived as
\begin{gather}
\frac{\partial{f_{u}}(\mathbf{p})}{\alpha_{u, i}} = \sum_{j=1}^{M} \hat{u}_{i}^{j} \mathcal{V}' \left( \sum_{t=1}^{N} p_{t} \hat{u}_{t}^{j} - u^{j} \right).
\label{partial1}
\end{gather}
For $\mathbf{p}$ and $\mathbf{q}$ in $\mathcal{A}_{u}$ we have:
\begin{gather}
\frac{\partial{f_{u}}(\mathbf{p})}{\alpha_{u, i}} - \frac{\partial{f_{u}}(\mathbf{q})}{\alpha_{u, i}} \nonumber \\
= \sum_{j=1}^{M} \hat{u}_{i}^{j} \left( \mathcal{V}' \left( \sum_{t=1}^{N} p_{t} \hat{u}_{t}^{j} - u^{j} \right) - \mathcal{V}' \left( \sum_{t=1}^{N} q_{t} \hat{u}_{t}^{j} - u^{j} \right) \right).
\label{partialdiff1}
\end{gather}
As $\mathcal{V}'$ is differentiable and contiuous in $\mathbb{R}$, based on Mean Value Theorem, there exists $r^{j}$ between $\sum_{t=1}^{N} p_{t} \hat{u}_{t}^{j} - u^{j}$ and $\sum_{t=1}^{N} q_{t} \hat{u}_{t}^{j} - u^{j}$ such that
\begin{gather}
\mathcal{V}'' \left( r^{j} \right) = \frac{\mathcal{V}' \left( \sum_{t=1}^{N} p_{t} \hat{u}_{t}^{j} - u^{j} \right) - \mathcal{V}' \left( \sum_{t=1}^{N} q_{t} \hat{u}_{t}^{j} - u^{j} \right)}{\sum_{t=1}^{N} u_{t}^{j} \left( p_{t} - q_{t} \right)}.
\label{MeanValueTheorem1}
\end{gather}
Note that $r^{j}$ is derived independently from $i$, and it is fixed for all $i$. Substituting (\ref{MeanValueTheorem1}) in (\ref{partialdiff1}), results
\begin{gather}
\frac{\partial{f_{u}}(\mathbf{p})}{\alpha_{u, i}} - \frac{\partial{f_{u}}(\mathbf{q})}{\alpha_{u, i}} = \sum_{t=1}^{N} \left( \sum_{j=1}^{M} \mathcal{V}''\left( r^{j} \right) \hat{u}_{i}^{j} \hat{u}_{t}^{j} \right) \left( p_{t} - q_{t}\right) 
\label{partialdiff2}
\end{gather}
%From strictly convexity of $\mathcal{V}$, $\mathcal{V}''(r^{j})>0$, thus: $\sqrt{\mathcal{V}''\left( r^{j}\right)} \in \mathbb{R}^{+}$. %Let us denote $\mathbf{D}_{M \times N}$ as the matrix with elements $D_{ji} = \sqrt{\mathcal{V}''\left( r^{j} \right)} \hat{u}_{i}^{j}$, $j=1, ..., M$ and $i=1, ..., N$. The matrix $\mathbf{Q} = \mathbf{D}^{T} \mathbf{D}$ is a semi-positive matrix with $Q_{it} = \sum_{j=1}^{M} \mathcal{V}'' \left(r^{j} \right) u_{i}^{j} u_{t}^{j}$. 
%Rewriting (\ref{partialdiff2}) in the form of matrix equation, we have:
%\begin{gather}
%\nabla f_{u} \left( \mathbf{p} \right) - \nabla f_{u} \left( \mathbf{q} \right) = \mathbf{D}^{T} \mathbf{D} \left( \mathbf{p} - \mathbf{q} \right) = \mathbf{Q} \left( \mathbf{p} - \mathbf{q} \right).
%\label{NablaDiff}
%\end{gather}

Using the inequality $0 < \mathcal{V}''(r^{j}) \leq B_u^{j}$ and taking the norm of (\ref{partialdiff2}) results in:
\begin{gather}
|| \nabla f_{u} \left( \mathbf{p} \right) - \nabla f_{u} \left( \mathbf{q} \right) || \nonumber \\
%N \max_{i, k} \left| D_{ik} \right| || \mathbf{p} - \mathbf{q} || \nonumber \\
%\leq N \max_{i, k} \left| \sum_{j=1}^{M} \hat{u}_{i}^{j} \hat{u}_{k}^{j} {\mathcal{V}_{u}^{j}} '' \left( \sum_{l=1}^{N} r^{i}_{l} \hat{u}_{l}^{j} - u^{j} \right) \right|  || \mathbf{p} - \mathbf{q} || \nonumber \\
\leq N \max_{i, k} \left \{ \sum_{j=1}^{M} \left| \hat{u}_{i}^{j} \hat{u}_{k}^{j} \right| B_{u}^{j}  \right \} || \mathbf{p} - \mathbf{q} || = N L^{\max} || \mathbf{p} - \mathbf{q}||.
\end{gather}

Based on Proposition 2.3.2 and equation (2.30) in \cite{B99}, for $0 < \beta_u < \frac{2}{N L^{\max}}$, $f_{u}(\boldsymbol{\alpha}_{u}^{k})$ decreases as $k$ increases:
\begin{gather}
f_u(\boldsymbol{\alpha}_{u}^{k+1})-f_{u}(\boldsymbol{\alpha}_{u}^{k}) \nonumber \\
\leq \left( \frac{N L^{\max}}{2} -\frac{1}{\beta_{u}}\right)||\boldsymbol{\alpha}_u^{k+1}-\boldsymbol{\alpha}_u^{k}||_{2}^{2}
\label{fk}
\end{gather}
As the right side of (\ref{fk}) is non-positive, $f_u(\boldsymbol{\alpha}_u^{n+1})=f_u(\boldsymbol{\alpha}_u^{n})$ if and only if $\boldsymbol{\alpha}_{u}^{n+1}=\boldsymbol{\alpha}_{u}^{n}$, $\exists n \geq 0$. In this case, $\boldsymbol{\alpha}_{u}^{k}=\boldsymbol{\alpha}_{u}^{n}$ for all $k \geq n$, and $\tilde{\boldsymbol{\alpha}}_{u}^{*}=\boldsymbol{\alpha}_{u}^{n}$ is a limit point of $\{\boldsymbol{\alpha}^{k}\}_{k=1}^{\infty}$. Thus, the sequence $\{f_u(\boldsymbol{\alpha}_{u}^{k})\}_{k=1}^{\infty}$ converges monotonically decreasing to $f_{u}^*=f_{u}(\tilde{\boldsymbol{\alpha}}^{*}_{u})$. The function $f_u(\alpha_u)$ is continuous in $\mathbb{R}^{N}$ and is bounded below by zero. Based on the monotone convergence Theorem \cite{T12}, the sequence $\{f_u(\boldsymbol{\alpha}_u^k)\}_{k=1}^{\infty}$ converges to a limit point $f_u^*=\lim_{k \rightarrow \infty}f_u(\boldsymbol{\alpha}_u^k)$.

For the only other possible case, $\nexists n \geq 0$ such that $\boldsymbol{\alpha}_{u}^{n} = \boldsymbol{\alpha}_{u}^{n+1}$. In this case, $f_{u}(\boldsymbol{\alpha}_{u}^{k})$ monotonically decreases as $k$ increases, i.e. $f_{u}(\boldsymbol{\alpha}_{u}^{n}) < f_{u}(\boldsymbol{\alpha}_{u}^{m})$, for $m>n \geq 0$. The function $f_u(\boldsymbol{\alpha}_k)$ is a non-negative, continuous function that is bounded below by zero. Based on monotone convergence Theorem, $\{f_u({\boldsymbol{\alpha}_u^{m}})\}_{m=1}^{\infty}$ converges to a limit point $f^{*}=\lim_{m \rightarrow \infty}f(\boldsymbol{\alpha}^{m})$. As $\mathcal{A}_u$ is a compact set, according to Bolzano-Weierstrass Theorem \cite{R76}, there exists an infinite subset $\{\tilde{\boldsymbol{\alpha}}_{u}^{k}\}_{k=1}^{\infty} \subseteq \{\boldsymbol{\alpha}_u^{m}\}_{m=1}^{\infty}$ that has a limit point in $\mathcal{A}_{u}$, i.e. $\exists \tilde{\boldsymbol{\alpha}}_u^* \in \mathcal{A}_{u}: \lim_{n \rightarrow \infty}\tilde{\boldsymbol{\alpha}}_u^{n}=\tilde{\boldsymbol{\alpha}}_u^*$. As $f_u(\boldsymbol{\alpha}_u)$ is continuous in $\mathcal{A}_u$, we have: $f^*_u=\lim_{m \rightarrow \infty} f(\boldsymbol{\alpha}_u^{m})=\lim_{m \rightarrow \infty}f(\tilde{\boldsymbol{\alpha}}_u^{m})=f(\tilde{\boldsymbol{\alpha}}^*)$.

In both cases, according to Proposition 2.3.2 in \cite{B99}, the limit point is stationary and satisfies the optimality condition: $\nabla f_u(\tilde{\boldsymbol{\alpha}}_u^{*})^{T}(\boldsymbol{\alpha} - \tilde{\boldsymbol{\alpha}}_{u}^{*}) \geq 0, \forall \boldsymbol{\alpha} \in \mathcal{A}_{u}$. From the convexity of $f_{u}(\boldsymbol{\alpha}_{u})$ in $\mathcal{A}_{u}$, Proposition 2.1.2 in \cite{B99} implies that the limit point $\boldsymbol{\alpha}_{u}^{*}$ minimizes $f_{u}(\boldsymbol{\alpha}_{u})$ over $\mathcal{A}_u$. 
%In order to show that $f_u(\boldsymbol{\alpha}^{k})$ monotically decreases, and converges to a limit point, $\tilde{\boldsymbol{\alpha}}_u^{*}$, in $\boldsymbol{\alpha}_u$ let consider two different cases
%As the right side of (\ref{fk}) is non-positive, $f_{u}(\boldsymbol{\alpha}_u^{n})=f_{u}(\boldsymbol{\alpha}_{u}^{n+1})$ if and only if $\boldsymbol{\alpha}^{n}=\boldsymbol{\alpha}^{n+1}$, for some $n \geq 1$. In this case, $\boldsymbol{\alpha}^{n}=\$, and therefore, $\tilde{\boldsymbol{\alpha}}^{*}=\boldsymbol{\alpha}^{n}$ is a limit point of the sequence $\{\boldsymbol{\alpha}^{k}\}_{k=1}^{\infty}$ 
\end{proof}
\subsection{Proof of Lemma \ref{Lem1}:}
\begin{proof}
Let us define the Lagrangian function
\begin{gather}
\mathcal{L}(\boldsymbol{\alpha}_u, \boldsymbol{\lambda})=\sum_{i=1}^{N}(c_{i,u}^{k}+\alpha_{i,u})^2 - 2 \lambda_{u}\left(\sum_{i=1}^{N}\alpha_{i,u}-1\right).
\label{Lagrange1}
\end{gather}
Referring to Proposition 3.4.1 in \cite{B99}, the feasible coefficient vector $\boldsymbol{\alpha}_{u}^{*}$ is an optimal solution for (\ref{Lagrange1}), if and only if there is the Lagrangian coefficient $\lambda_u^{*} \in \mathbb{R}$ such that $\boldsymbol{\alpha}_u^{*}$ minimizes $\mathcal{L}(\boldsymbol{\alpha}_u, \lambda_u^*)$ over $(\mathbb{R}^{+})^{N}$  
\begin{gather}
\boldsymbol{\alpha}_{u}^{*}=\arg\min_{\boldsymbol{\alpha}_{u} \geq \mathbf{0}}\mathcal{L}\left(\boldsymbol{\alpha}_u, \lambda_u^*\right) \nonumber \\
=\arg\min_{\boldsymbol{\alpha}_u \geq \mathbf{0}}\sum_{i=1}^{N}\left(c_{i,u}^{k}+\alpha_{i,u} - \lambda_{u}^{*}\right)^2 \nonumber \\
=\sum_{i=1}^{N}\arg\min_{\alpha_{i,u} \geq 0}\left(c_{i,u}^{k}+\alpha_{i,u} - \lambda_{u}^{*}\right)^2.
\label{Lagrange2}
\end{gather}
The optimal solution for (\ref{Lagrange2}) is derived by minimizing each term $(c_{i,u}^{k}+\alpha_{i,u} - \lambda_{u}^{*})^2$ over $\alpha_{i,u} \in \mathbb{R}^{+}$ as
\begin{equation}
\alpha_{i,u}^{*}=[\lambda_u^* - c_{i,u}^{k}]^{+}, \hspace{5mm} i = 1, ..., N.
\end{equation}
Let us denote $g(\lambda_u)=\sum_{i=1}^{N}[\lambda_u - c_{i,u}^{k}]^{+}-1$. This function is continuous and strictly increasing in $\mathbb{R}$, and we have $g(\min_{i}c_{i,u}^{k})=-1$, and $g(\min_{i}c_{i,u}^{k} + 1) \geq 0$. Therefore, there is only a unique $\lambda_u^* \in \Lambda$ that satisfies the constraint  $g(\lambda_u^{*}) = \sum_{i=1}^{N}\alpha_{i,u}^{*}-1 = 0$.
\end{proof}
\subsection{Proof of Lemma \ref{Lem2}:}
\begin{proof}
As the function $y(n) = \sum_{i=1}^{n - 1}[m_n - m_i] - 1$ is strictly increasing in $n \in \{1, ..., N + 1\}$, and we have: $y(1) = -1$ and $y(N + 1) \geq 0$, there exists a unique value $n$ such that $\sum_{i=1}^{n - 1}[m_n - m_i] - 1 \leq 0$ and $\sum_{i=1}^{n}[m_{n + 1} - m_i] - 1 \geq 0$. As the function $g(\lambda_u)=\sum_{i=1}^{N}[\lambda_u - c_{i,u}^{k}]^{+} - 1$ is continuous and strictly increasing in $\mathbb{R}$, it concludes that $m_{n} \leq \lambda_u^{*} \leq m_{n+1}$. We have:
\begin{gather}
0=g(\lambda_u^*)=\sum_{i=1}^{N}[\lambda^{*}_u - c_{i,u}^{k}]^{+} - 1=\sum_{i=1}^{n}\left(\lambda^{*}_u - m_{i}\right) - 1.
\label{lu}
\end{gather}
Thus, we have: $\lambda_{u}^{*} = (1+\sum_{i=1}^{n}m_{i})/n$.
\end{proof}

\subsection{Proof of Theorem \ref{Thm2}:}
\begin{proof} As the generalized distance function $\mathcal{V}_{u}^{j}(x)$ is strictly convex in $\mathbb{R}$, $f_{u}\left( \boldsymbol{\alpha}_{u} \right)$ in (\ref{falpha}) is convex in $\mathcal{A}_{u}$. Let us denote $\boldsymbol{\alpha}_{u}^{*}$ as an optimal solution for (\ref{falpha}). Note that as $\boldsymbol{\alpha}_{u}^{*} \in \mathcal{A}_{u}$ minimizes $f_{u}(\boldsymbol{\alpha}_{u})$, from the first order inequality constraint, we have:
\begin{gather}
\nabla f(\boldsymbol{\alpha}_{u}^{*})^{T} \left( \boldsymbol{\alpha}_{u} - \boldsymbol{\alpha}^{*} \right) \geq 0, \hspace{5mm} \forall \boldsymbol{\alpha}_{u} \in \mathcal{A}_{u}.
\label{FirstOrder}
\end{gather}
For $\beta_{u}>0$, the first order inequality in (\ref{FirstOrder}) is equivalent to:
\begin{gather}
\left( \left( \boldsymbol{\alpha}_{u}^{*} - \beta_{u} \nabla f(\boldsymbol{\alpha}^{*}_{u}) \right) - \boldsymbol{\alpha}_{u}^{*} \right)^{T} \left( \boldsymbol{\alpha}_{u} - \boldsymbol{\alpha}^{*} \right) \leq 0, \hspace{5mm} \forall \boldsymbol{\alpha}_{u} \in \mathcal{A}_{u}.
\label{AppDProjection}
\end{gather}
Using Proposition 2.1.3 (b) in \cite{B99}, the inequality in (\ref{AppDProjection}) holds if and only if $\boldsymbol{\alpha}_{u}^{*}$ is the projection of $\left( \boldsymbol{\alpha}_{u}^{*} - \beta_{u} \nabla f(\boldsymbol{\alpha}^{*}_{u} \right )$ on $\mathcal{A}_{u}$, i.e.:
\begin{gather}
\boldsymbol{\alpha}_{u}^{*} = \mathcal{P}_{\mathcal{A}_{u}}\left[ \boldsymbol{\alpha}_{u}^{*} - \beta_{u} \nabla f_{u} \left( \boldsymbol{\alpha}_{u}^{*} \right) \right], \hspace{5mm} \forall \beta_{u}>0.
\label{AppDProjection2}
\end{gather}
From convexity of $\mathcal{A}_{u}$, for any point $\mathbf{A}$ and $\mathbf{B}$ in $\mathbb{R}^{N}$, the distance between projection of $\mathbf{A}$ and $\mathbf{B}$ on $\mathcal{A}_{u}$, respectively denoted as $\mathcal{P}_{\mathcal{A}_{u}}[\mathbf{A}]$ and $\mathcal{P}_{\mathcal{A}_{u}}[\mathbf{B}]$, is less or equal to the distance between $\mathbf{A}$ and $\mathbf{B}$:
\begin{gather}
|| \mathcal{P}_{\mathcal{A}_{u}}\left[ \mathbf{A} \right] - \mathcal{P}_{\mathcal{A}_{u}}\left[ \mathbf{B} \right] || \leq ||\mathbf{A} - \mathbf{B}||, \hspace{5mm} \forall \mathbf{A}, \mathbf{B} \in \mathbb{R}^{N}.
\label{ProjectionDistanceInequality}
\end{gather}
Substitutting $\boldsymbol{\alpha}_{u}^{k}, \boldsymbol{\alpha}_{u}^{*} \in \mathcal{A}_{u}$ in (\ref{partialdiff2}), we have:
\begin{gather}
\nabla f_{u}\left( \boldsymbol{\alpha}_{u}^{k}\right) - \nabla f_{u}\left( \boldsymbol{\alpha}_{u}^{*} \right) = \mathbf{D}^{k} \left( \boldsymbol{\alpha}_{u}^{k} - \boldsymbol{\alpha}_{u}^{*} \right).
\end{gather}
where $D_{im}^{k}=\frac{\partial^{2} f_{u}(\mathbf{r}^{i})}{\partial \alpha_{u,i} \partial \alpha_{u, m}}=\sum_{j=1}^{M} \hat{u}^{j}_{i}\hat{u}^{j}_{m} {\mathcal{V}_{u}^{j}} '' \left( \sum_{l=1}^{N} r_{l}^{i} \hat{u}_{l}^{j} - u^{j} \right)$. From $0 < l^{\min} \leq D_{ij} \leq L^{\max}$, $i = 1, .., N$ and $0 < \beta_{u} < \frac{2}{N L^{\max}}$, we have:
\begin{gather}
-1 < 1 - \beta_{u} N L^{\max} \leq 1 - \beta_{u} \sum_{m=1}^{N}D^{k}_{im} \nonumber \\
\leq 1 - \beta_{u} N l^{\min} < 1,
\label{Thm2LminLmax}
\end{gather}
or
\begin{gather}
|1 - \beta_{u} \sum_{m=1}^{N}D^{k}_{im}| \nonumber \\
 \leq q = \max\{ |1 - \beta_{u} N l^{\min}|, |1 - \beta_{u} N L^{\max}| \} < 1.
\end{gather}
Following the same argument in (2.34) in \cite{B99}, substituting from (\ref{alphak}) and (\ref{AppDProjection2}), and using (\ref{ProjectionDistanceInequality}), we have:
\begin{gather}
||\boldsymbol{\alpha}_{u}^{k+1} - \boldsymbol{\alpha}_{u}^{*}|| = ||\mathcal{P}_{\mathcal{A}_{u}}\left[ \boldsymbol{\alpha}_{u}^{k} - \beta_{u} \nabla f_{u} \left( \boldsymbol{\alpha}_{u}^{k}\right)\right] \nonumber \\
- \mathcal{P}_{\mathcal{A}_{u}}\left[ \boldsymbol{\alpha}_{u}^{*} - \beta_{u} \nabla f_{u} \left( \boldsymbol{\alpha}_{u}^{*} \right) \right]|| \nonumber \\
\leq ||\left( \boldsymbol{\alpha}_{u}^{k} - \beta_u \nabla f_{u} \left( \boldsymbol{\alpha}_{u}^{k} \right) \right) - \left( \boldsymbol{\alpha}_{u}^{*} - \beta_u \nabla f_{u} \left( \boldsymbol{\alpha}_{u}^{*} \right) \right) || \nonumber \\
= || \left( I - \beta_{u} \mathbf{D}^{k}) \right) \left( \boldsymbol{\alpha}_{u}^{k} - \boldsymbol{\alpha}_{u}^{*} \right) || \nonumber \\
\leq \max\{|1 - \beta_{u} N l^{\min}|, |1 - \beta_{u} N L^{\max}| \} ||\boldsymbol{\alpha}_{u}^{k} - \boldsymbol{\alpha}_{u}^{*}|| \nonumber \\
= q ||\boldsymbol{\alpha}_{u}^{k} - \boldsymbol{\alpha}_{u}^{*}||.
\end{gather}
Thus, the update process convergences with linear rate of $q$. 

From properties of $\boldsymbol{\alpha}_{u} \in \mathcal{A}_{u}$ in (\ref{Cons}), and $0 \leq \alpha_{u,i} \leq 1$, $i = 1, ..., N$, we have:
\begin{gather}
\sum_{i=1}^{N} \alpha_{u,i}^{2} \leq \sum_{i=1}^{N} \alpha_{u, i} = 1, \hspace{5mm} \alpha_{u} \in \mathcal{A}_{u}.
\label{alphaSquareSum}
\end{gather}
Starting from $\boldsymbol{\alpha}_{u}^{0}=\left( \frac{1}{N}, ..., \frac{1}{N} \right)$, using (\ref{alphaSquareSum}), we have:
\begin{gather}
||\boldsymbol{\alpha}_{u}^{k+1} - \boldsymbol{\alpha}_{u}^{*}||^{2}_2 \leq q^{2k} ||\boldsymbol{\alpha}_{u}^{0} - \boldsymbol{\alpha}_{u}^{*} ||^{2}_2 \nonumber\\
= q^{2k} \sum_{i=1}^{N} \left( \frac{1}{N} - \boldsymbol{\alpha}_{u, i}^{*} \right)^{2} = q^{2k} \left(\sum_{i=1}^{N} \left( \alpha_{u, i}^{*} \right)^{2} - \frac{1}{N} \right) \nonumber \\
\leq q^{2k} \left( 1 - \frac{1}{N}\right)
\label{norm2q2k}
\end{gather}
In order to guarantee that the update process in (\ref{alphak}) converges to the $\epsilon$-neighborhood of the optimal solution, starting from $\boldsymbol{\alpha}_{u}^{0}= \left( \frac{1}{N}, ..., \frac{1}{N} \right)$, using (\ref{norm2q2k}), we have:
\begin{gather}
||\boldsymbol{\alpha}_{u}^{k+1} - \boldsymbol{\alpha}_{u}^{*}||_{2} \leq q^{k} \sqrt{\left( 1 - \frac{1}{N}\right)} < \epsilon.
\end{gather}
In order that the solution of the update process at step $k+1$ lies in the $\epsilon$-neighborhood of $\boldsymbol{\alpha}_{u}^{*}$, it is sufficient that:
\begin{equation}
k > \log_{\left( \frac{1}{q} \right)} \frac{ \sqrt{ 1 - \frac{1}{N} } }{\epsilon} = \frac{\ln \frac{\sqrt{1 - \frac{1}{N}}}{\epsilon}}{\ln \frac{1}{q}},
\end{equation}
where $\ln(x)$ is the logarithm of $x$ in the natural base $e$.

As discussed in Section \ref{OptimalProjectionAlgorithm}, finding optimal solution of each update process using Lemma \ref{Lem2} takes $O(N \log N)$ number of calculations. Thus, running the update process for $\log_{\left( \frac{1}{q} \right)} \frac{ \sqrt{ 1 - \frac{1}{N} } }{\epsilon}$ times requires $O\left(N \log \left( N \right) \log_{\left( \frac{1}{q} \right)} \frac{ \sqrt{ 1 - \frac{1}{N} } }{\epsilon} \right)$ total number of calculations.
\end{proof}

% \bibliographystyle{ieeetr}
% \bibliography{RFID}

\end{document}